\documentclass[AMA,STIX1COL]{WileyNJD-v2}
\usepackage{moreverb}

\usepackage{hypernat}
\usepackage{natbib}
\setcitestyle{super}
\usepackage{showframe}
\newcommand\BibTeX{{\rmfamily B\kern-.05em \textsc{i\kern-.025em b}\kern-.08em
T\kern-.1667em\lower.7ex\hbox{E}\kern-.125emX}}

\usepackage{amsmath,amsthm,amssymb}
\usepackage{graphicx,epstopdf,framed}
\usepackage{euscript,paralist}
\usepackage{float}

\numberwithin{equation}{section}

\newfont{\roma}{cmr10 scaled 1200}

\newcommand{\rline}  {{\mathbb R}}

\newcommand{\dd}   {{\rm d}\hbox{\hskip 0.5pt}}

\newcommand{\mm}    {{\hbox{\hskip 0.5pt}}}

\newcommand{\bluff} {{\hbox{\raise 15pt \hbox{\mm}}}}
\newcommand{\sbluff}{{\hbox{\raise  7pt \hbox{\mm}}}}

\newcommand{\FORALL} {{\hbox{$\hskip 11mm \forall \;$}}}

\renewcommand{\theequation}{{\arabic{section}.\arabic{equation}}}

\newcommand{\bbm}[1]{\left[\begin{matrix} #1 \end{matrix}\right]}

\usepackage{setspace}

\articletype{Research Article}%

\received{:}
\revised{:}
\accepted{:}

\begin{document}

\title{Model reference adaptive control for state and input constrained linear systems}

\author[1]{Sudipta Chattopadhyay}

\author[1]{Srikant Sukumar*}

\author[1]{Vivek Natarajan}

\authormark{Sudipta Chattopadhyay \textsc{et al}}

\address[1]{\orgdiv{Systems and Control Engineering}, \orgname{Indian Institute of Technology Bombay}, \orgaddress{\state{Mumbai}, \country{India}}}

\corres{*Srikant Sukumar, Systems and Control Engineering, Indian Institute of Technology Bombay, Powai, Mumbai-400076, India. \email{srikant.sukumar@iitb.ac.in}}

\presentaddress{Systems and Control Engineering, Indian Institute of Technology Bombay, Powai, Mumbai-400076, India}

\abstract[Summary]{State and input constraints are ubiquitous in all engineering systems. In this article, we derive adaptive controllers for uncertain linear systems under pre-specified state and input constraints. Several modifications of the model reference adaptive control (MRAC) framework have been proposed to address input constraints in uncertain linear systems. Considering the infeasibility of arbitrary reference trajectories, reference modification has been implemented in the case of input constraints in literature. The resulting conditions on the reference and input signals are difficult to verify online. Similar results on state and input constraints together have also been proposed, albeit resulting in more complex and unverifiable conditions on the control. The primary objective of this article is therefore to account for state and input constraints in uncertain linear systems by providing easily verifiable conditions on the control and reference. A combination of reference modification and barrier Lyapunov methods in adaptive control are employed to arrive at these results. }



 \maketitle

\section{Introduction}
\label{sec1}
\setcounter{equation}{0} 

Most physical and chemical systems are designed to operate within certain safety limits and these limits translate to state and input constraints associated with the plant dynamics. To ensure safe and reliable operation of these systems, it is important to design control strategies that take into account these safety constraints. In addition, it is common for system parameters to be unknown in applications, necessitating the use of adaptation and identification.

Several efforts have been made to address the problem of designing controllers for linear systems that can handle state or/and input constraints. Methods such as model predictive control \cite{DeDhBh:23, BeBoMo:02},  robust optimal control \cite{MaSc:97}, invariant set theory \cite{Bl:99}, reference governor methods  \cite{BeCaMo:97}, control barrier function based quadratic programming \cite{AmGrTa:14} have been proposed for designing controllers that can handle various system constraints. However, all these methods involve solving an optimization problem in real-time that is computationally expensive. Further, for these methods, obtaining an analytic closed-form control law and proving stability under such safety constraints, is often very difficult. Feasibility of pre-specified constraints is also not easily answered in the optimal control context. Handling uncertainties (see \cite{AdGu:11,DhBh:21,LoCaAl:19,Ma:14} and references therein) is also possible, to an extent, though adding to the complexity of the associated optimization problem. 

In this article, we focus on control design methods {for uncertain linear systems} with proven stability guarantees, as opposed to optimality. For uncertain linear systems, a computationally inexpensive control method is the model reference adaptive control (MRAC) \cite[Chapter 6]{IoSu:12} that generates control actions such that an uncertain system tracks the behavior of a given stable reference model. MRAC controllers come with strong stability guarantees but do not guarantee system operation within pre-defined state and input constraints. Therefore, designing a modified MRAC, such that state and input constraints are respected, is a problem of practical interest.  

Several articles (see \cite{Af:18, ZhLo:22, ArYuBa:20} and references therein) propose MRACs that guarantee the convergence of the system tracking errors to zero while satisfying the state or error constraint for all time. In recent times, barrier Lyapunov function (BLF) based controller design has gained prominence for design of controllers, including MRACs that can handle state constraints (see \cite{AmGrTa:14, Af:18, KoAm:19} and references therein). Although BLF-based controllers guarantee system operation within user-defined state or error constraints, these controllers often result in large control effort whenever the states or errors approach their safety limits, potentially violating the plant input constraints. At the other end of the spectrum, for control constraints, a modified MRAC with reference trajectory modification has been proposed in \cite{LaHo:07} which has it's origins in \cite{KaAn:93}.

While there exists sufficient literature on input or state constraints individually, there are very few articles that consider these constraints simultaneously. As is evident in the BLF framework from our previous discussion, state and input constraints present contradictory challenges in adaptive control. In \cite{GhBh:22}, the authors have proposed an MRAC which guarantees system operation within user-defined bounds on state and input. In \cite{AnMaAf:21}, the design of an MRAC that can handle user-prescribed state and input constraints is achieved by developing an auxiliary reference model and using barrier Lyapunov functions. However, in \cite{GhBh:22} and \cite{AnMaAf:21}, the authors have assumed the existence of a feasible control policy such that their results hold and do not provide any verifiable condition that guarantees the existence of such a control policy. Obtaining such a verifiable condition is crucial in practical implementations. 

In this context, articles such as \cite{LaHo:07} have had far-reaching impact, since they do provide analytical conditions on system states for successful tracking under uncertain parameters and control constraints. The scope of their applications is also wider than pure BLF based methods such as in \cite{AnMaAf:21} due to modification of the reference that is allowed in the framework of \cite{LaHo:07}. However, the conditions presented in \cite{LaHo:07} have proven to be rather difficult to verify offline or online. Motivated by reference modification based methods and their subsequent works, the authors in \cite{WaLiYaTi:22} have demonstrated reference modification to handle \emph{both} state and control constraints. This however leads to implicit conditions that are even more complex to verify in practice than those in \cite{LaHo:07} due to the two layers of reference modification involved.


The contributions of our article arise from two main questions posed as an outcome of the above discussion.

\begin{enumerate}
    \item Is it possible to construct adaptive controllers that account for simultaneous state and control constraints? - We answer this in the affirmative and develop an adaptive controller with state constraints handled by BLF based design and control constraints via reference modification. Rigorous stability proofs are provided in the MRAC framework. This is different from \cite{AnMaAf:21} that use only BLF based design and \cite{WaLiYaTi:22} which use only reference modifications.

    \item  Is it possible to formulate verifiable conditions that quantify the feasible constraint sets? - We derive conditions on the pre-specified control and state norm bounds, that are easily verifiable online and also prior to implementation. This is a significant improvement over \cite{WaLiYaTi:22} and references therein where reference modification results in unverifiable conditions on control and state bounds. 
\end{enumerate}

We further demonstrate extensions of our results to additive Lipschitz nonlinearities in this article. We illustrate using example simulations that the combination of BLF and reference modification based design allows for consideration of more stringent state and control constraints while facilitating easy verification of feasible constraint combinations.


The rest of the paper is organized as follows: Section \ref{sec2} contains the problem statement. In Section \ref{sec3}, we design a tracking controller which guarantees that the state constraint will be satisfied for all time. In Section \ref{sec4}, we propose a modification to the reference trajectory and thereby to the control law such that both the state and input constraints remain satisfied for all time. In Section \ref{sec5}, we demonstrate the efficacy of our proposed controller and reference trajectory modification using a numerical example. Finally, we present some concluding remarks and directions for future work in Section \ref{sec6}.

\noindent
{\it Notations:} $\|X\|_2$ denotes the Euclidean norm of a vector $X \in \rline^n$. ${\rm Sign}(x)$ represents the sign of the variable $x$ and therefore is $1$ if $x  \geq 0$ and $-1$ if $x < 0$. $I_n \in \rline^{n \times n}$ represents the identity matrix of order $n$. The smallest and largest eigenvalues of a symmetric matrix $A \in \rline^{n\times n}$ are denoted as $\lambda_{\min}(A)$ and $\lambda_{\max}(A)$ respectively.

\section{Problem Statement} \label{sec2} \setcounter{equation}{0} 
Consider a linear time-invariant system whose dynamics are governed by the following state space model: $\forall \ t \geq 0$
\begin{equation}\label{eq:sys}
\dot{X}(t) = AX(t)+B\lambda u(t),\ \ \ X(0) = X_0.
\end{equation}
Here $u(t) \in \rline$ is the input and $X(t) \in \rline^n$ is the state. The system matrix $A \in \rline^{n \times n}$ and the input gain $\lambda \in \rline$ are  assumed to be unknown but the sign of $\lambda$ is considered to be known. The pair $(A,B\lambda)$ is assumed to be controllable and full state measurements are assumed to be available. Let the target system  be governed by the following state space model:
\begin{equation}\label{eq:target}
\dot{X}_m(t) = A_mX_m(t)+B_mf(t),\ \ \ X_m(0) = X_{m_0},
\end{equation}
where $A_m\in\rline^{n\times n}$ is the target system matrix,  $X_m(t) \in \rline^n$ is the target system state and the reference input $f$ is a scalar valued, bounded, continuous function of time. We consider the following assumptions on the system in \eqref{eq:sys} and on the target system in \eqref{eq:target}:
\begin{assumption}\label{as:match} \textit{There exist a vector $K\in\rline^n$ and a nonzero scalar $l \in \rline$ such that $A_m-A = B\lambda K^\top$ and $B_m = B\lambda l$ hold. Further, there exist known positive scalars $M_K$, $M_l$ and $m_l$ such that $K$ and $l$ satisfy $\|K\|_2 \leq M_K$ and $0<m_l \leq |l| \leq M_l$.}
\end{assumption}
\begin{assumption}\label{as:Lyap} \textit{The target system matrix $A_m$ is Hurwitz i.e., for any given symmetric positive definite matrix $Q$, there exists a symmetric positive definite matrix $P$ such that the algebraic Lyapunov equation $A_m^\top P + P A_m = -Q$ is satisfied.}
\end{assumption}

While the first statement in Assumption \ref{as:match} is similar to the usual matching condition in the MRAC literature (see \cite{LaHo:07} and \cite{IoSu:12} Chapter 6), the second statement imposes a requirement of some prior knowledge about the gains $K$ and $l$. While the knowledge of bounds $M_K$, $m_l$ and $M_l$ is not required and can be arbitrary for the control implementation, these are critical for apriori verification of constraint feasibility (see Remarks \ref{rm:Naira_diff} and \ref{rm:tighter_bounds}). Recall from the discussion below \eqref{eq:sys} that $B$, $B_m$ and ${\rm Sign}(\lambda)$ are known. Therefore, it is easy to see from $B_m = B\lambda l$ in Assumption \ref{as:match} that ${\rm Sign}(l)$ is also known. Now consider the following constraints under which we want to solve a tracking problem for the system in \eqref{eq:sys}:

\textbf{State constraint:} For a given positive real constant $M_x$, the state vector $X$ in \eqref{eq:sys} should satisfy
\begin{equation}\label{eq:state_constraint}
\|X(t)\|_2 < M_x \FORALL  t \geq 0. 
\end{equation}

\textbf{Input constraint:} For a given positive real constant $M_u$, the system input $u$ in \eqref{eq:sys} should satisfy 
\begin{equation}\label{eq:input_constraint}
 |u(t)|\leq M_u \FORALL t \geq 0.
\end{equation}

Now we state the problem of interest in this paper.
 
\begin{problem}\label{prob}
\textit{Design an adaptive control input $u$ for the system in \eqref{eq:sys} such that $X$ in \eqref{eq:sys} tracks the target system state $X_m$ in \eqref{eq:target} while $X(t)$ and $u(t)$ satisfy the state and input constraints \eqref{eq:state_constraint}-\eqref{eq:input_constraint} for all $t \geq 0$.}
\end{problem}


\begin{remark}
In a general setting with arbitrary $M_x,\,M_u$ in \eqref{eq:state_constraint}-\eqref{eq:input_constraint}, a controller solving Problem~\ref{prob} may not exist. For this reason, reference trajectory modification has been proposed in the literature (see \cite{LaHo:07, WaLiYaTi:22} and references therein) to handle the input constraints. In this method, the reference trajectory $f(t)$ of the target system is modified in such a way that the system state $X$ can track the modified target system states while satisfying the input constraint. We utilize a similar approach in this article. The key novelty in this article as opposed to literature on MRAC under input constraints is the inclusion of state constraints. We also present more easily verifiable conditions on the viable constraints as opposed to those in literature.
\end{remark}


In Section \ref{sec3}, we design an adaptive control input $u$ such that the system state $X$ can track the target system state $X_m$ while only the state constraint remains satisfied for all $t$. In Section \ref{sec4}, we modify the control law proposed in Section \ref{sec3} to account for both the input and the state constraint.


\section{MRAC under state constraint}\label{sec3}\setcounter{equation}{0} 

In this section, we design a control input $u$ to solve Problem \ref{prob} for the system in \eqref{eq:sys} when $M_u$ is arbitrarily large. In particular, we provide a result which shows that our designed control input drives the tracking error to $0$ as $t$ goes to infinity while the state constraint in \eqref{eq:state_constraint} remains satisfied for all $t \geq 0$. Throughout this section, we suppose that the reference trajectory $f(t)$, the initial reference state $X_{m_0}$ and the matrices $A_m$ and $B_m$ in \eqref{eq:target} are chosen in such a way that the target system state $X_m$ satisfies
\begin{equation}\label{eq:xmi_MXM}
\|X_m(t)\|_2 \leq M_{x_m} < M_x\ \ \ \forall \ t \geq 0
\end{equation}
for some positive constant $M_{x_m}$. Note that this results in no loss of generality since if $M_{x_m} > M_x$ then it is not possible for $X$ in \eqref{eq:sys} to simultaneously track $X_m$ in \eqref{eq:target} and satisfy the state constraints given in \eqref{eq:state_constraint}.

 Denote the tracking error for the system state $X$ in \eqref{eq:sys} as 
\begin{equation}\label{eq:E}
    E(t) = X(t)-X_m(t) \FORALL t \geq 0.
\end{equation}
Using this, \eqref{eq:sys}, \eqref{eq:target} and the matching conditions from Assumption \ref{as:match}, we get the error dynamics as follows:
\begin{equation}\label{eq:Edot}
\dot{E}(t) = A_m E(t) - B\lambda{K}^\top\! X(t) +B\lambda u(t)- B_mf(t),
\end{equation}
where $K$ and $l$ are as in Assumption \ref{as:match}. Since $A$ and $\lambda$ are unknown, we have $K$ and $l$ also unknown. We denote the estimates for $K$ and $l$ at time $t$ as $\hat K(t)$ and $\hat l(t)$ respectively and the errors of estimation for $K$ and $l$ at time $t$ as $\tilde{K}(t) = \hat K(t) - K$ and $\tilde{l}(t) = \hat l(t) - l$. Define $M_e = M_x - M_{x_m}$, where $M_{x_m}$ is as in \eqref{eq:xmi_MXM}. Since $M_{x_m} < M_x$, we have $M_e > 0$. Now we present a useful intermediate result.

\begin{lemma}\label{lm:Mtoxm}
\textit{Let the bounds in \eqref{eq:xmi_MXM} hold and $M = M_e \sqrt{\lambda_{\rm min}(P)}$, where $P$ is as in Assumption \ref{as:Lyap}. Then we have the following result:
\begin{equation}
E^\top(t)PE(t) < M^2\ \  \forall\ t \geq 0 \implies  \|X(t)\|_2 < M_x \ \forall \ t \geq 0.
\end{equation}}
\end{lemma}
\begin{proof}
Observe from Assumption \ref{as:Lyap} that $P$ is a positive definite matrix and therefore we have $\lambda_{\rm min}(P) > 0$. Further, employing the definition of $M$ in the lemma statement we get the following: 
$$E^\top(t)PE(t) < M^2 \implies \lambda_{\rm min}(P)E^\top(t)E(t) < \lambda_{\rm min}(P)M_e^2, \,  \forall\ t \geq 0 .$$
 This further implies that $\|E(t)\|_2 < M_e$ for all $t \geq 0$ since $\lambda_{\rm min}(P) > 0$. Now, it is easy to see from this, \eqref{eq:xmi_MXM} and the definition of $M_e$ that 
\begin{equation}\label{eq:etox_bound}
\|X(t)\|_2 \leq \|E(t)\|_2+\|X_m(t)\|_2 < M_e+ M_{x_m} = M_x
\end{equation}
for all $t \geq 0$. This completes the proof.
\end{proof}

Denote 
\begin{equation}\label{eq:Lambda}
\mu(t) = \frac{2M^2E^\top(t) P B_m{\rm Sign}(l)}{(M^2 - E^\top(t) P E(t))^2}.
\end{equation}
\begin{theorem}\label{th:state_cons}
\textit{Let the bounds in \eqref{eq:xmi_MXM} and Assumptions \ref{as:match}, \ref{as:Lyap} hold. Let $M$ be as in Lemma \ref{lm:Mtoxm} and $\Gamma_K$ and $\Gamma_l$ be some positive constants. If the input to the system in \eqref{eq:sys} is given as
\begin{equation}\label{eq:inputs}
u(t) = \hat{K}^\top\!(t) X(t)+\hat l(t) f(t),
\end{equation} 
where the projection based update laws for $\hat l(t)$ and $\hat K (t)$ are given as}

\begin{equation}\label{eq:Lupdt}
  \!\!\!\!  \dot{\hat{l}}(t) \!=\! 
\begin{cases}
    -\Gamma_l\mu(t) f(t),& \hspace{0mm} \textit{if } 0<m_l <|\hat{l}(t)| < M_l\\
            & \hspace{0mm} \textit{or} \ |\hat{l}(t)| = m_l\ \textit{and}\ \! \mu(t) f(t)\hat{l}(t) \leq 0\\
            & \hspace{0mm} \textit{or} \ |\hat{l}(t)| = M_l \ \textit{and}\ \!\mu(t) f(t)\hat{l}(t) \geq 0\\
    0,              & \textit{if} \ |\hat{l}(t)| = m_l\ \textit{and}\ \! \mu(t) f(t)\hat{l}(t) > 0\\
            &  \textit{or} \ |\hat{l}(t)| = M_l \ \textit{and}\ \!\mu(t) f(t)\hat{l}(t) < 0
\end{cases}
\end{equation}
\begin{equation}\label{eq:Kupdt}
   \!\!\dot{\hat{K}}(t) \!= \!
\begin{cases}
    -\Gamma_K\mu(t) X(t),& \hspace{0mm} \textit{if } \|\hat{K}(t)\|_2 < M_K\\
            & \hspace{0mm} \textit{or} \ \|\hat{K}(t)\|_2 = M_K\ \textit{and}\ \mu(t) X^\top\!(t) \hat K(t) \geq 0\vspace{2mm} \\ 
    -\Gamma_K(I - \frac{\hat{K}(t)\hat{K}^\top(t)}{\hat{K}^\top\!(t)\hat{K}(t)})\mu(t) X(t),              & \textit{if} \ \|\hat{K}(t)\|_2 = M_K\ \textit{and}\ \mu(t) X^\top\!(t) \hat K(t) < 0
\end{cases}
\end{equation}
 \textit{then we obtain the following  provided that $\|\hat{K}(0)\|_2 \leq M_K$, $m_l \leq |\hat{l}(0)| \leq M_l$ and $E^\top\!(0) P E(0) < M^2$ hold:} 
\begin{enumerate} 
\item \textit{$\|\hat K(t)\|_2 \leq M_K$ and $ m_l \leq |\hat l(t)| \leq M_l$ for all $t \geq 0$,}
\item \textit{$\|X(t)\|_2 < M_x$ for all $t \geq 0$,}
\item \textit{the adaptive system in \eqref{eq:Edot} has bounded solutions and $E(t) \to 0$ as $t \to \infty$.}
\end{enumerate}
\end{theorem}


\textit{(Proof of statement 1:)} Since we have $m_l  \leq |\hat{l}(0)| \leq M_l$, it is easy to see from \eqref{eq:Lupdt} that $m_l \leq |\hat l (t)| \leq M_l$ for all $t \geq 0$. Therefore, we only need to show that $\|\hat K (t)\|_2 \leq M_K$ for all $t \geq 0$. Consider the following positive semi-definite function $V_2(t)$:
\begin{equation}\label{eq:V2}
V_2(t) = \frac{1}{2}\hat{K}^\top(t)\hat{K}(t).
\end{equation}
Differentiating $V_2$ with respect to time and using \eqref{eq:Kupdt} we get
$$\dot{V}_2(t) = - \Gamma_K\mu(t) X^\top\!(t)(I - \frac{\hat{K}(t)\hat{K}^\top\!(t)}{\hat{K}^\top\!(t)\hat{K}(t)})\hat{K}(t) = 0$$ 
when $\|\hat{K}(t)\|_2 = M_K$ and $ \mu(t) X^\top\!(t) \hat K(t) < 0$. Further, when $\|\hat{K}(t)\|_2 = M_K$ and $ \mu(t) X^\top\!(t) \hat K(t) \geq 0$ we get
$$\dot{V}_2(t) = - \Gamma_K\mu(t) X^\top\!(t)\hat{K}(t)\leq 0.$$
From these and the definition of $V_2$, we can conclude that $\|\hat K(t)\|_2 \leq M_K$ for all $t \geq 0$ provided $\|\hat{K}(0)\|_2 \leq M_K$. This completes the proof of (1) in the theorem statement.
\qed

\textit{(Proof of statement 2:)} Define $\gamma = {\rm Sign}(l)/l$ and recall $M$ from Lemma \ref{lm:Mtoxm}. Consider the following barrier Lyapunov function:
\begin{equation}\label{eq:lyap}
V(t) = \frac{E^\top\!(t) P E(t)}{M^2 - E^\top\!(t) P E(t)} + \frac{1}{2 \Gamma_K}\tilde{K}^\top\!(t) \gamma\tilde{K}(t) + \frac{1}{2 \Gamma_l}\gamma\tilde{l}(t)^2,
\end{equation}
where $\Gamma_K$ and $\Gamma_l$ are positive constants and the positive definite matrix $P$ is as in Assumption \ref{as:Lyap}. Observe from the statement of this theorem that since  $E^\top\!(0) P E(0) < M^2$ and $\gamma, \Gamma_K, \Gamma_l > 0$, $V(0)$ is non negative. Further, if $E^\top\!(t) P E(t) < M^2$ for all $t \geq 0$ then $V$ is continuous in time $t$. Differentiating $V$ with respect to time $t$, we get
\begin{equation*}
\dot{V}(t) =  \frac{M^2 \frac{\dd }{\dd t}(E^\top\!(t) P E(t))}{(M^2 - E^\top\!(t) P E(t))^2} + \frac{\gamma}{{\Gamma_K}}\dot{\hat{K}}^\top\!(t) \tilde{K}(t)+ \frac{\gamma}{{\Gamma_l}}\tilde{l}(t)\dot{\hat{l}}(t).
\end{equation*}
Using \eqref{eq:Edot}, \eqref{eq:inputs}, the matching conditions in Assumption \ref{as:match} and the algebraic Lyapunov equation in Assumption \ref{as:Lyap}, we get the following after a simple calculation:
\begin{align}
&\dot{V}(t) =  \frac{-M^2 E^\top\!(t) Q E(t) }{(M^2 - E^\top\!(t) P E(t))^2}+\Big(\frac{\mu(t)f(t)}{l\ {\rm Sign}(l)}+\frac{\gamma}{{\Gamma_l}}\dot{\hat{l}}(t)\Big)\tilde{l}(t)\nonumber\\&+\Big(\frac{\mu(t)X^\top\!(t)}{l\ {\rm Sign}(l)} + \frac{\gamma}{{ \Gamma_K}}\dot{\hat{K}}^\top\!(t)\Big) \tilde{K}(t), \label{eq:Vdot2} 
\end{align}
where $\mu(t)$ is as defined in \eqref{eq:Lambda}. Substituting for the update law for $\hat{K}$ given in \eqref{eq:Kupdt}, we get the following:

\begin{equation*}
   \!\!\!\!\!\!\!\!\! \Big(\frac{\mu(t)X^\top\!(t)}{l\ {\rm Sign}(l)} + \frac{\gamma}{{ \Gamma_K}}\dot{\hat{K}}^\top\!(t)\Big) \tilde{K}(t)
\begin{cases}
    =0,& \text{if } \|\hat{K}(t)\|_2 < M_K\\
            & \text{or} \ \|\hat{K}(t)\|_2 = M_K\ \text{and}\ \mu(t) X^\top\!(t) \hat K(t) \geq 0\vspace{2mm} \\ 
    \leq 0,              & \text{if} \ \|\hat{K}(t)\|_2 = M_K\ \text{and}\ \mu(t) X^\top\!(t) \hat K(t) < 0
\end{cases}
\end{equation*}
provided that $\|{K}\|_2 \leq M_K$ and $\|\hat{K}(0)\|_2 \leq M_K$. Similarly, substituting for the update law for $\hat{l}$ given in \eqref{eq:Lupdt}, we get that
\begin{equation*}
  \!\!\!\!  \Big(\frac{\mu(t)f(t)}{l\ {\rm Sign}(l)}+\frac{\gamma}{{ \Gamma_l}}\dot{\hat{l}}(t)\Big)\tilde{l}(t) 
\begin{cases}
    =0,&  \text{if } 0<m_l <|\hat{l}(t)| < M_l\\
            &  \text{or} \ |\hat{l}(t)| = m_l\ \text{and}\ \! \mu(t) f(t)\hat{l}(t) \leq 0\\
            &  \text{or} \ |\hat{l}(t)| = M_l \ \text{and}\ \!\mu(t) f(t)\hat{l}(t) \geq 0\\
    \leq 0,              & \text{if} \ |\hat{l}(t)| = m_l\ \text{and}\ \! \mu(t) f(t)\hat{l}(t) > 0\\
            &  \text{or} \ |\hat{l}(t)| = M_l \ \text{and}\ \!\mu(t) f(t)\hat{l}(t) < 0
\end{cases}
\end{equation*}
provided that $m_l  \leq |{l}| \leq M_l$ and $m_l  \leq |\hat{l}(0)| \leq M_l$. Using these in \eqref{eq:Vdot2} and the fact that $Q$ is a positive definite matrix, we get
\begin{equation}\label{eq:finalVdot}
\dot{V}(t) \leq -\frac{M^2 E^\top\!(t) Q E(t) }{(M^2 - E^\top\!(t) P E(t))^2} \leq 0 \ \ \ \forall  \ t \geq 0
\end{equation}
since $\|\hat K (t)\|_2 \leq M_K$ and $m_l \leq |\hat l (t)| \leq M_l$ hold for all $t \geq 0$ (see statement 1 of Theorem \ref{th:state_cons}).

Observe from \eqref{eq:lyap} and \eqref{eq:finalVdot} that since $E^\top(0) P E(0) < M^2$ and $E$ is a continuous function of time $t$, $E^\top(t) P E(t) < M^2$ for all $t \geq 0$. Otherwise if $E^\top(t) P E(t) = M^2$ at some point of time $t$ then we get $V(t)$ at that time $t$ to be infinity. This along with $E^\top(0) P E(0) < M^2$, $\|K\|_2 \leq M_K$, $\|\hat{K}(0)\|_2 \leq M_K$, $m_l \leq |l| \leq M_l$ and $m_l  \leq |\hat{l}(0)| \leq M_l$ contradicts \eqref{eq:finalVdot}. Therefore, we obtain
\begin{equation}\label{eq:Ebound}
E^\top(t) P E(t) < M^2\FORALL t \geq 0
\end{equation}
which implies $\|X(t)\|_2 < M_x$ for all $t \geq 0$ (see Lemma \ref{lm:Mtoxm}) since the bounds in \eqref{eq:xmi_MXM} hold. This completes the proof of (2) in the theorem statement. 
\qed

\textit{(Proof of statement 3:)} Observe from \eqref{eq:lyap}, \eqref{eq:finalVdot} and \eqref{eq:Ebound} that $\dot{V}(t) = 0$ only if $E^\top\!(t) QE(t) = 0$. So, using Barbalat's Lemma, signal chasing analysis and the positive definiteness of the matrix $Q$ we can show that $E(t) \to 0$ asymptotically. Further, since $P$ is a constant positive definite matrix we get from \eqref{eq:Ebound} that the adaptive system in \eqref{eq:Edot} has bounded solutions. This completes the proof of statement (3) in the theorem. \qed


\section{MRAC under state and input constraints}\label{sec4}\setcounter{equation}{0} 

In this section, we propose a modification to the reference trajectory $f$ of the target system in \eqref{eq:target} (and thereby modification of $u$ in \eqref{eq:inputs}) to handle \emph{both} the input and state constraints present in \eqref{eq:state_constraint}-\eqref{eq:input_constraint}. In particular, this reference modification keeps the modified control input saturated at the boundary values $\pm M_u$ whenever the absolute value of the control input given in \eqref{eq:inputs} exceeds the bound $M_u$. We also provide a verifiable condition (see \eqref{eq:stability_con}) that guarantees the stability of the closed-loop system dynamics of the plant in \eqref{eq:sys} along with the corresponding modified input in the presence of constraints.

Recall the plant state $X$ from \eqref{eq:sys}, $l$ from Assumption \ref{as:match}, $M_x$ from \eqref{eq:state_constraint} and the gains $\hat{K}$ and $\hat{l}$ from \eqref{eq:inputs}. We first present the assumptions under which we prove our main theorem in this section.
\begin{assumption}\label{as:target}
\textit{The matrices $A_m$ and $B_m$ and the initial state $X_{m_0}$ in the target system in \eqref{eq:target} are designed in such a way that for any given positive constant $M_{x_m}$ satisfying $M_{x_m} < M_x$ there exists a positive constant $f_M$ such that for any reference trajectory $f$ satisfying $|f(t)|\leq f_M$ for all $t \geq 0$, we have $\|X_m(t)\|_2 \leq M_{x_m} < M_x$ for all $t \geq 0$.} 
\end{assumption}
\begin{assumption}\label{as:stability_cond}
\textit{Let $f_M$ be as in Assumption \ref{as:target}. The following inequality holds for all $t \geq 0$:
\begin{equation}\label{eq:stability_con}
M_u \geq |\hat{K}^\top\!(t) X(t)|-\hat{l}(t){\rm Sign}(l)f_M
\end{equation}}
\end{assumption}

Assumption \ref{as:target} is without loss of generality since the existence of a positive constant $f_M$ for any given $M_{x_m}$, such that Assumption \ref{as:target} holds, is always guaranteed for any exponentially stable system (which is the case for our target system). This is because bounded inputs will result in bounded states for such systems. Furthermore, using tools like Lyapunov analysis, it is also easy to find a $f_M$ for a given $M_{x_m}$ such that Assumption \ref{as:target} holds. Assumption \ref{as:stability_cond} guarantees the stability of the closed-loop system dynamics of the plant in \eqref{eq:sys} along with the corresponding modified input in the presence of constraints. Assumption \ref{as:stability_cond} belongs to a class of assumptions standard in constrained MRAC literature (see Remark \ref{rm:Naira_diff}) and can be verified online in our case unlike existing research.

\begin{remark}\label{rm:Naira_diff}
Modified MRACs that can handle state or actuator constraints require additional conditions (in terms of the constraint parameters) to ensure the stability of the closed-loop system dynamics of the plant. Literature in constrained MRAC (see \cite{LaHo:07, KaAn:93, WaLiYaTi:22} and references therein) are invariably accompanied by stability conditions such as Assumption \ref{as:stability_cond} (which provides a relation between state and control bounds). However, the stability conditions presented in \cite{LaHo:07, KaAn:93, WaLiYaTi:22} etc., are complicated and difficult to verify due to the presence of \emph{true values} of the unknown parameters in those conditions. In a significant improvement over the existing literature on constrained MRACs, all terms on the right-hand side of \eqref{eq:stability_con} are either measurable or known.
Further, with information on bounds of $\|\hat K\|_2$, $|\hat l|$ and $\|X\|_2$, it is possible to verify \eqref{eq:stability_con} in Assumption \ref{as:stability_cond} even before implementing our control algorithm (see Remark \ref{rm:tighter_bounds}).
\end{remark}

Consider the following modified target system obtained from \eqref{eq:target}:
\begin{equation}\label{eq:modtarget}
\dot{X}_m^s(t) = A_mX_m^s(t)+B_m(f(t)+g(t)),\ X_m^s(0) = X_{m_0},
\end{equation}
where $X_m^s \in \rline^n$ is the state of the modified target system, $g(t) \in \rline$ is an additive reference modification term and the continuous bounded reference input $f$ and the matrices $A_m$ and $B_m$ are as in Section \ref{sec2}. A desirable characteristic for $g$ is that it should be continuous, uniformly bounded and $g(t) = 0$ whenever $|u(t)| \leq M_u$. This is to prevent any modification of the original reference trajectory $f(t)$ when the input constraint is already satisfied. 

Denote the modified tracking error of the system state $X$ in \eqref{eq:sys} at time $t$ as 
\begin{equation}\label{eq:Es}
E^s(t) = X(t)-X_m^s(t).
\end{equation}
Recall $M_x$ and suppose that Assumption \ref{as:target} holds. Fix an $M_{x_m}$ satisfying $M_{x_m} < M_x$ and let $f_M$ be as in Assumption \ref{as:target}. Then we define $M_e = M_x - M_{x_m}$ and $M = M_e \sqrt{\lambda_{\rm min}(P)}$, where $P$  is as in Assumption \ref{as:Lyap}. Let $\mu^s(t)$ be defined as
\begin{equation}\label{eq:Lambda_s}
\mu^s(t) = \frac{2M^2{E^s}^\top(t) P B_m{\rm Sign}(l)}{(M^2 - {E^s}^\top(t) P E^s(t))^2}.
\end{equation}
Now, we present a useful intermediate result. 
\begin{lemma}\label{lm:sec4}
\textit{Suppose Assumptions \ref{as:match}, \ref{as:Lyap} and \ref{as:target} hold, and the modified reference trajectory $f+g$ in \eqref{eq:modtarget} is continuous and satisfies $|f(t)+g(t)|\leq f_M$. {Let $\Gamma_K$ and $\Gamma_l$ be some positive constants.} If the input to the system in \eqref{eq:sys} is given as
\begin{equation}\label{eq:mod_input}
u^s(t) =  {\hat{K}}^\top\!(t) X(t) + \hat l(t) f(t) + \hat l(t) g(t).
\end{equation}
where the projection-based update laws for $\hat K$ and $\hat l$ are given as 
\begin{equation}\label{eq:Ksupdt}
   \!\! \dot{\hat{K}}(t) \!= \!
\begin{cases}
    -{ \Gamma_K}\mu^s(t) X(t),& \hspace{0mm} \text{if } \|\hat{K}(t)\|_2 < M_K\\
       & \hspace{0mm} \text{or} \ \|\hat{K}(t)\|_2 = M_K\ \text{and}\ \mu^s(t) X^\top\!(t) \hat K(t) \geq 0\vspace{2mm} \\ 
    -{ \Gamma_K}(I - \frac{\hat{K}(t)\hat{K}^\top(t)}{\hat{K}^\top\!(t)\hat{K}(t)})\mu^s(t) X(t),         & \text{if} \ \|\hat{K}(t)\|_2 = M_K\ \text{and}\ \mu^s(t) X^\top\!(t) \hat K(t) < 0
\end{cases}
\end{equation}
\begin{equation}\label{eq:Lsupdt}
  \!\! \dot{\hat{l}}(t) \!=\! 
\begin{cases}
    -{ \Gamma_l}\mu^s(t) (f+g)(t),& \hspace{0mm} \text{if } 0<m_l <|\hat{l}(t)| < M_l\\
            & \hspace{0mm} \text{or} \ |\hat{l}(t)| = m_l\ \text{and}\ \! (\mu^s (f+g)\hat{l})(t) \leq 0\\
            & \hspace{0mm} \text{or} \ |\hat{l}(t)| = M_l \ \text{and}\ \!(\mu^s (f+g)\hat{l})(t) \geq 0\\
    0,              & \hspace{0mm} \text{if} \ |\hat{l}(t)| = m_l\ \text{and}\ \! (\mu^s (f+g)\hat{l})(t) > 0\\
            & \hspace{0mm} \text{or} \ |\hat{l}(t)| = M_l \ \text{and}\ \!(\mu^s (f+g)\hat{l})(t) < 0
\end{cases}
\end{equation}
then we obtain the following provided that $\|\hat{K}(0)\|_2 \leq M_K$, $m_l \leq |\hat{l}(0)| \leq M_l$ and ${E^s}^\top\!(0) P E^s(0) < M^2$ hold: 
\begin{enumerate} 
\item $\|\hat K(t)\|_2 \leq M_K$ and $ m_l \leq |\hat l(t)| \leq M_l$ for all $t \geq 0$,
\item $\|X(t)\|_2 < M_x$ for all $t \geq 0$,
\item $E^s(t) \to 0$ as $t \to \infty$.
\end{enumerate}}
\end{lemma}

\begin{proof}
Mimicking the steps followed in the proof of Theorem \ref{th:state_cons}, we can prove statements 1 and 3 of this lemma. Further, we can obtain that ${E^s}^\top P E^s < M^2$ for all $t \geq 0$. This implies using the definition of $M$ that  
\begin{equation}\label{eq:sec4lm1}
\|E^s(t)\|_2 < M_e \ \ \forall \ t \geq 0.
\end{equation}
Again, since Assumption \ref{as:target} holds and the modified reference trajectory $f(t)+g(t)$ in \eqref{eq:modtarget} satisfies $|f(t)+g(t)|\leq f_M$, we get from Assumption \ref{as:target} that 
\begin{equation}\label{eq:sec4lm2}
\|X^s_{m}(t)\|_2 \leq M_{x_m} < M_x\ \ \forall \ t \geq 0.
\end{equation}
Now it is easy to see from \eqref{eq:Es}, \eqref{eq:sec4lm1}, \eqref{eq:sec4lm2} and the definition of $M_e$ present below \eqref{eq:Es} that 
$\|X(t)\|_2 < M_x$ for all $t \geq 0$ (for further details see Lemma \ref{lm:Mtoxm}). This completes the proof of statement 2 of the lemma.
\end{proof}



Now, we state our main theorem.

\begin{theorem}\label{th:state_ip_cons}
\textit{Let Assumptions \ref{as:match} and \ref{as:Lyap} hold. Let $A_m$, $B_m$, $X_{m_0}$, $f_M$, $M_{x_m}$ and $M_u$ be such that Assumptions \ref{as:target} and \ref{as:stability_cond} hold. Suppose $f$ in \eqref{eq:modtarget} is continuous and satisfies $|f(t)| \leq f_M$ for all $t \geq 0$. Let the input to the system in \eqref{eq:sys} be given by $u^s$ in \eqref{eq:mod_input} where the update laws for $\hat K$ and $\hat l$ are given in \eqref{eq:Ksupdt}-\eqref{eq:Lsupdt} and $g$ is as follows:}
\begin{equation}\label{eq:g}
  \!\! g(t) \!=\! 
\begin{cases}
    0,&  \textit{if } |u(t)| < M_u\\
          \frac{M_u-u(t)}{\hat{l}(t)},  &  \textit{if } u(t) \geq M_u\\
           
          \frac{-M_u-u(t)}{\hat{l}(t)}  &  \textit{if } u(t) \leq -M_u\\
\end{cases}
\end{equation}
\textit{where} 
\begin{equation}\label{eq:thp:u}
u(t) = \hat{K}^\top\!(t) X(t) + \hat l(t) f(t).
\end{equation}
\textit{Then the following holds provided that $\|\hat{K}(0)\|_2 \leq M_K$, $m_l  \leq |\hat{l}(0)| \leq M_l$ and ${E^s}^\top(0) P E^s(0) < M^2$:}
\begin{enumerate}
\item \textit{$|f(t)+g(t)|\leq f_M$ for all $t \geq 0$,}
\item \textit{$\|\hat K(t)\|_2 \leq M_K$ and $ m_l \leq |\hat l(t)| \leq M_l$ for all $t \geq 0$,}
\item \textit{$X(t) \to X_m^s(t)$ as $t \to \infty$,}
\item  \textit{$\|X(t)\|_2 < M_x$ and $|u^s(t)| \leq M_u$ for all $t \geq 0$.}
\end{enumerate}
\textit{Furthermore, if $M_u$ goes to $\infty$, we get that $\limsup_{t \to \infty} \|\tilde E(t)\|_2 \to 0$, where $\tilde E(t) = E(t)-E^s(t)$ and $E(t)$ is as in \eqref{eq:E}.}
\end{theorem}

\textit{(Proof of statement 1) :} Note from \eqref{eq:Lsupdt} that $|\hat{l}|$ is lower bounded by a positive scalar $m_L$ and is therefore never $0$. Using this and \eqref{eq:g} we get that $g$ is a continuous function of $t$ and  
\begin{align*}
f(t)+g(t) &= f(t), \hspace{24mm} {\rm if} \ |u(t)| <M_{u}\\
&= \frac{M_{u}- \hat{K}^\top\!(t) X(t)}{\hat{l}(t)}, \ \ \  {\rm if} \ \ u(t) \ \geq M_{u}\\
&= \frac{-M_{u}- \hat{K}^\top\!(t) X(t)}{\hat{l}(t)}, \  {\rm if} \ \ u(t)\ \leq -M_{u}.
\end{align*} 
Since we have by design $|f(t)| \leq f_M$ for all $t \geq 0$, we get that $|f(t)+g(t)|\leq f_M$ whenever $|u(t)| < M_u$. Therefore, we need to prove the same only for those cases when $|u(t)| \geq M_u$. Suppose the sign of $l$ is positive. Then $\hat l \geq m_l > 0$ and the following holds: 
\begin{enumerate}
\item Suppose $u(t) \geq M_u$. Then we have $\hat{K}^\top\!(t) X(t) + \hat{l}(t)f(t) \geq M_u$ which implies $M_u - \hat{K}^\top\!(t) X(t) \leq \hat l(t) f(t) \leq \hat l(t) f_M$. This further implies $f(t)+g(t)\leq f_M$ whenever $u(t) \geq M_u$ and $\hat l > 0$.
\item  From \eqref{eq:stability_con}, we get $M_u \geq \hat{K}^\top\!(t) X(t)-\hat{l}(t)f_M$. This implies $f(t)+g(t)\geq -f_M$ whenever $u(t) \geq M_u$ and $\hat l > 0$. Again  from \eqref{eq:stability_con}, we get $M_u \geq -\hat{K}^\top\!(t) X(t)-\hat{l}(t)f_M$ which implies $f(t)+g(t)\leq f_M$ whenever $u(t) \leq -M_u$ and $\hat l > 0$.
\item Suppose $u(t) \leq -M_u$. Then we have $\hat{K}^\top\!(t) X(t) + \hat{l}(t)f(t) \leq -M_u$ which implies $-M_u - \hat{K}^\top\!(t) X(t) \geq \hat l(t) f(t) \geq -\hat l(t) f_M$. This further implies $f(t)+g(t)\geq -f_M$ whenever $u(t) \leq -M_u$ and $\hat l > 0$.
\end{enumerate}

For the case when ${\rm Sign}(l) < 0$ similar arguments as above can be used to claim that $|f(t)+g(t)|\leq f_M$.
From these results, we get that $|f(t)+g(t)|\leq f_M$ for all $t \geq 0$. This completes the proof of statement 1 of the theorem.
\qed

\textit{(Proof of statements 2,3 and 4) :} Mimicking the steps followed in the proof of Theorem \ref{th:state_cons}, we can easily prove statements (2) and (3) of this theorem and that ${E^s}^\top(t)PE^s(t) < M^2$ for all $t \geq 0$ since $\|K\|_2 < M_K$, $\|\hat{K}(0)\|_2 \leq M_K$, $m_l < |l| < M_l$, $m_l  \leq |\hat{l}(0)| \leq M_l$ and ${E^s}^\top(0) P E^s(0) < M^2$ hold. Since we have from statement (1) of Theorem \ref{th:state_ip_cons} that $|f(t)+g(t)|\leq f_M$ for all $t \geq 0$ and $f+g$ is a continuous function of time $t$, we get from Lemma \ref{lm:sec4} that $\|X(t)\|_2 < M_x$ for all $t \geq 0$.



Using \eqref{eq:mod_input}, \eqref{eq:g} and the expression of $u(t)$ given in \eqref{eq:thp:u}, we get 
\begin{equation}\label{eq:input_sat}
u^s(t) = 
\begin{cases}
    u(t),& \hspace{-1mm} \text{if } \ |u(t)|<M_{u} \\
    M_u {\rm Sign}(u(t)),  & \hspace{-1mm} \text{if } \ |u(t)|\geq M_{u}. 
\end{cases}
\end{equation} 
This implies that $|u^s(t)| \leq M_u$ for all $t \geq 0$. This completes the proof of (4) in the theorem statement. \qed

\textit{(Proof of last statement) :} Using \eqref{eq:target}, \eqref{eq:modtarget} and the definition of $\tilde E$ we get that
\begin{equation}\label{eq:Endot}
\dot{\tilde E}(t) = A_m\tilde E(t)+B_mg(t).
\end{equation} 
Observe from the expression of $u$ in \eqref{eq:thp:u} that $|u(t)| \leq M_KM_x+M_lf_M < C$ for some constant $C>0$ and for all $t \geq 0$. Clearly, from this and \eqref{eq:g} we get that $|g(t)| \leq \min(2f_M,\max((C-M_u)/m_l,0)) < \infty$ for all $t \geq 0$. This implies that since $C, m_l$ and $f_M$ are fixed positive constants, $\sup_{t\in[0,\infty)}|g(t)|$ goes to 0 as $M_u$ goes to infinity. This along with \eqref{eq:Endot} further implies that if $M_u$ goes to $\infty$, we get that $\limsup_{t \to \infty} \|\tilde E(t)\|_2 \to 0$. This completes the proof of the theorem. \qed

\begin{remark}\label{rm:diff_not}
 In the absence of the reference trajectory modification proposed in \eqref{eq:g}, we can still obtain a lower bound of $M_u$ directly from \eqref{eq:inputs}. In particular, we can obtain $M_u \geq $$|\hat{K}^\top\!(t) X(t)+\hat{l}(t)f(t)|$ uniformly for all time, $t$ from \eqref{eq:inputs}. Clearly, the lower bound on $M_u$ obtained using this approach is greater than the lower bound for $M_u$ given in \eqref{eq:stability_con}. Therefore, the class of reference trajectories and constraints that we can handle using our proposed method is larger than the class that can be handled without using the reference modification.
\end{remark}

\begin{remark}\label{rm:tighter_bounds}
Observe in Theorem \ref{th:state_ip_cons} that $\|\hat{K}(t)\|_2 \leq M_K$, $m_l \leq |\hat l(t)|$ and $\|X(t)\|_2 \leq M_x$ for all $t\geq 0$. We can use these bounds to verify Assumption \ref{as:stability_cond} prior to implementing our control algorithm. In particular, Assumption \ref{as:stability_cond} can be verified by checking whether $M_u \geq  M_K M_x - m_lf_M$. Verifying Assumption \ref{as:stability_cond} using these bounds is more conservative than online verification. Bounds as obtained using $M_u \geq  M_K M_x - m_lf_M$ can however be used by a system designer for actuator specifications.
\end{remark}

\subsection{Effect of an additional nonlinear term in \eqref{eq:sys}}\label{sec4.1}

Consider a system whose dynamics are governed by the following equation: $\forall \ t \geq 0$
\begin{equation}
\label{eq:non_Sys}
\dot{X}(t) = AX(t) +A_1 \Phi(X(t)) + B\lambda u(t), \ \ X(0)=X_0,
\end{equation}
where $A$, $B$, $\lambda$ and $u$ are as in Section \ref{sec2}, $A_1 \in \rline^{n\times n}$ is unknown and $\Phi(X(t)) \in \rline^n$ is a measurable vector whose each element is a globally Lipschitz continuous function of the state vector $X$. In this section, the goal is to address Problem \ref{prob} for the system in \eqref{eq:non_Sys}. First, we consider the following assumptions:
\begin{assumption}\label{as:match_non} \textit{There exist a vector $K_1\in\rline^n$ such that $A_1 = -B\lambda K_1^\top$ holds. Further, there exists  a known positive scalar $M_{K_1}$ such that $K_1$ satisfies $\|K_1\|_2 \leq M_{K_1}$.}
\end{assumption}

\begin{assumption}\label{as:stablity_non}
\textit{Let $f_M$ be as in Assumption \ref{as:target}. The following inequality holds for all $t \geq 0$:
\begin{equation}\label{eq:stability_con_non}
M_u \geq |\hat{K}^\top\!(t) X(t)+\hat{K_1}^\top\!(t) \Phi(X(t))|-\hat{l}(t){\rm Sign}(l)f_M
\end{equation}}
\end{assumption}
Let the modified tracking error of the system state $X$ in \eqref{eq:non_Sys} at time $t$ be denoted as 
\begin{equation}\label{eq:Es_non}
E^s(t) = X(t)-X_m^s(t)
\end{equation}
where $X_m^s$ is as given in \eqref{eq:modtarget}. Recall $M_x$ in \eqref{eq:state_constraint} and $M_u$ in \eqref{eq:input_constraint}. Let $M_{x_m}$, $M_e$, $M$ and $\mu^s$ be defined as above Lemma \ref{lm:sec4} in Section \ref{sec4}. Then we get the following result:
\begin{theorem}\label{th:state_ip_cons_non}
\textit{Let Assumptions \ref{as:match}, \ref{as:Lyap} and \ref{as:match_non} hold. Let $A_m$, $B_m$, $X_{m_0}$, $f_M$, $M_{x_m}$ and $M_u$ be such that Assumptions \ref{as:target} and \ref{as:stablity_non} hold. Suppose $f$ in \eqref{eq:modtarget} is continuous and satisfies $|f(t)| \leq f_M$ for all $t \geq 0$. Let the input to the system in \eqref{eq:non_Sys} be given by 
$$u^s(t) =  {\hat{K}}^\top\!(t) X(t) + {\hat{K_1}}^\top\!(t) \Phi(X(t))+ \hat l(t) f(t) + \hat l(t) g(t)$$ 
where the update laws for $\hat K$ and $\hat l$ are given in \eqref{eq:Ksupdt}-\eqref{eq:Lsupdt}, the update law for $\hat K_1$ is given as }
\begin{equation}\label{eq:K1supdt}
   \!\!\!\!\!\!\!\!\! \dot{\hat{K_1}}(t) \!= \!
\begin{cases}
    -{ \Gamma_{K_1}}\mu^s(t) \Phi(X(t)),& \hspace{-0mm} \textit{if } \|\hat{K}_1(t)\|_2 < M_{K_1}\\
            & \hspace{-0mm} \textit{or} \ \|\hat{K}_1(t)\|_2 = M_{K_1}\ \\ 
            &\hspace{-0mm}\textit{and}\ \mu^s(t) \Phi(X(t))^\top \hat K_1(t) \geq 0\vspace{2mm} \\ 
    -{ \Gamma_{K_1}}(I - \frac{\hat{K}_1(t)\hat{K}_1^\top(t)}{\hat{K}_1^\top\!(t)\hat{K}_1(t)})\mu^s(t) \Phi(X(t)),              & \textit{if} \ \|\hat{K}_1(t)\|_2 = M_{K_1}\ \\ 
            &\hspace{-0mm}\textit{and}\ \mu^s(t) \Phi(X(t))^\top \hat K_1(t) < 0
\end{cases}
\end{equation}
\textit{ { for some positive constant $\Gamma_{K_1}$} and $g$ is as follows:}
\begin{equation}\label{eq:g_non}
  \!\! g(t) \!=\! 
\begin{cases}
    0,&  \textit{if } |u(t)| < M_u\\
          \frac{M_u-u(t)}{\hat{l}(t)},  &  \textit{if } u(t) \geq M_u\\
           
          \frac{-M_u-u(t)}{\hat{l}(t)}  &  \textit{if } u(t) \leq -M_u\\
\end{cases}
\end{equation}
\textit{ where} 
\begin{equation}\label{eq:thp:u_non}
u(t) = \hat{K}^\top\!(t) X(t)+ {\hat{K_1}}^\top\!(t) \Phi(X(t)) + \hat l(t) f(t).
\end{equation}
\textit{ Then the following holds provided that $\|\hat{K}(0)\|_2 \leq M_K$, $\|\hat{K}_1(0)\|_2 \leq M_{K_1}$, $m_l  \leq |\hat{l}(0)| \leq M_l$ and ${E^s}^\top(0) P E^s(0) < M^2$: }
\begin{enumerate}
 \item \textit{$|f(t)+g(t)|\leq f_M$ for all $t \geq 0$,}
\item \textit{$\|\hat K(t)\|_2 \leq M_K$, $\|\hat K_1(t)\|_2 \leq M_{K_1}$ and $ m_l \leq |\hat l(t)| \leq M_l$ for all $t \geq 0$,}
\item \textit{$X(t) \to X_m^s(t)$ as $t \to \infty$,}
\item \textit{$\|X(t)\|_2 < M_x$ and $|u^s(t)| \leq M_u$ for all $t \geq 0$.}
\end{enumerate}
\end{theorem}
The proof of Theorem \ref{th:state_ip_cons_non} follows from the proof of Theorem \ref{th:state_ip_cons}.
\section{Numerical example}\setcounter{equation}{0}
\label{sec5}

In this section, we demonstrate the efficacy of our control law proposed in Section \ref{sec4} using an example. In particular, we show that the system state $X$ in \eqref{eq:sys} can closely track the target system state $X_m$ in \eqref{eq:target} while maintaining the state and the input constraints for all time. The components of $X$ and $X_m$ are denoted by $x_i$ and $x_{m_i}$ respectively.

\noindent\textbf{Example:} Consider a third-order linear system whose state space model is given below:
\begin{equation}\label{eq:ex2_sys}
\!\!\!\!\dot{X}(t) \!=\! A\!X(t) + B\lambda u(t), \ \ X(0) = X_0
\end{equation}
where\vspace{-1mm}
$$A = \bbm{-0.5 & 1 & 1.85\\
  -1.2 &-1.7 &  -0.6\\
   2.5 & 0 & -0.4}, \ \ B = \bbm{0.5\\0\\ 1},$$
$\lambda = 0.5$ and $X_0 = \bbm{0.3 & -0.2 & 0.2}^\top$. Since at least one of the eigenvalues of $A$ has a positive real part, the above system is unstable. Further, in this example, the pair $(A,B\lambda)$ is controllable. Suppose $A$ and $\lambda$ are unknown but the sign of $\lambda$ is known.  We consider the following parameters for the state and input constraints given in \eqref{eq:state_constraint} and \eqref{eq:input_constraint}: $M_u = 3$ and $M_x = 2$.

Let the target system be as follows: $\forall \ t\geq 0$
\begin{equation}\label{eq:ex2_tar}
\dot{X}_m(t) = A_mX_m(t) + B_mf(t), \ \ X_m(0) = X_{m_0},
\end{equation}
where 
$$A_m = \bbm{-2 &  1.5 & 1.1\\
  -1.2 &-1.7 & -0.6\\
   -0.5 & 1 & -1.9}, \ \ B_m = \bbm{0.5\\0\\ 1}$$
and $X_{m_0} = \bbm{0.3 & -0.2 & 0.2}^\top$. Since all the eigenvalues of $A_m$ have negative real parts, the above target system is exponentially stable. 

Now, we verify the Assumptions \ref{as:match}, \ref{as:Lyap} and \ref{as:target}. Clearly, Assumption \ref{as:match} holds with $l = 2$ and $K = \bbm{-6& 2 &-3}^\top$ and since $A_m$ is Hurwitz, Assumption \ref{as:Lyap} also holds. Further, in Assumption \ref{as:match}, we consider $M_K = 10$, $m_l = 1$ and $M_l = 4$ (note that these bounds can be arbitrary). For the target system in \eqref{eq:ex2_tar} if we choose $f_M = 2.4$ and $M_{x_m} = 1.9$, then Assumption \ref{as:target} holds with these values of $M_{x_m}$ and $f_M$. Therefore, in \eqref{eq:ex2_tar}, we choose $f(t) = 1.4\sin(2t)+\sin(2.5t)$ which satisfies $|f(t)|\leq f_M$ for all $t \geq 0$. For this example, we verify Assumption \ref{as:stability_cond} online during the experiment.

From the definitions of $M_e$ and $M$ in Section \ref{sec4}, we get $M_e = 0.1$ and $M = 0.0474$. Next, we implement our control law given in \eqref{eq:mod_input}-\eqref{eq:Lsupdt} and the proposed reference modification given in \eqref{eq:g} with $\hat K (0) = \bbm{0.1 & 0.1 & 0.1}^\top$ , $\hat l (0) = 3$, { $\Gamma_K = 1$ and $\Gamma_l = 0.05$}. While implementing our proposed control law,  we modify the reference trajectory $f$ in \eqref{eq:ex2_tar} as $f+g$ and denote the modified target system state as $X_m^s$. We further denote $E^s$ as $E^s = X - X_m^s$ and the components of $X_m^s$ as $x_{m_i}^s$. We obtain the following results.\vspace{-8mm}
\begin{figure}[!htbp]
\centering{%
\resizebox*{9.3cm}{!}{\includegraphics{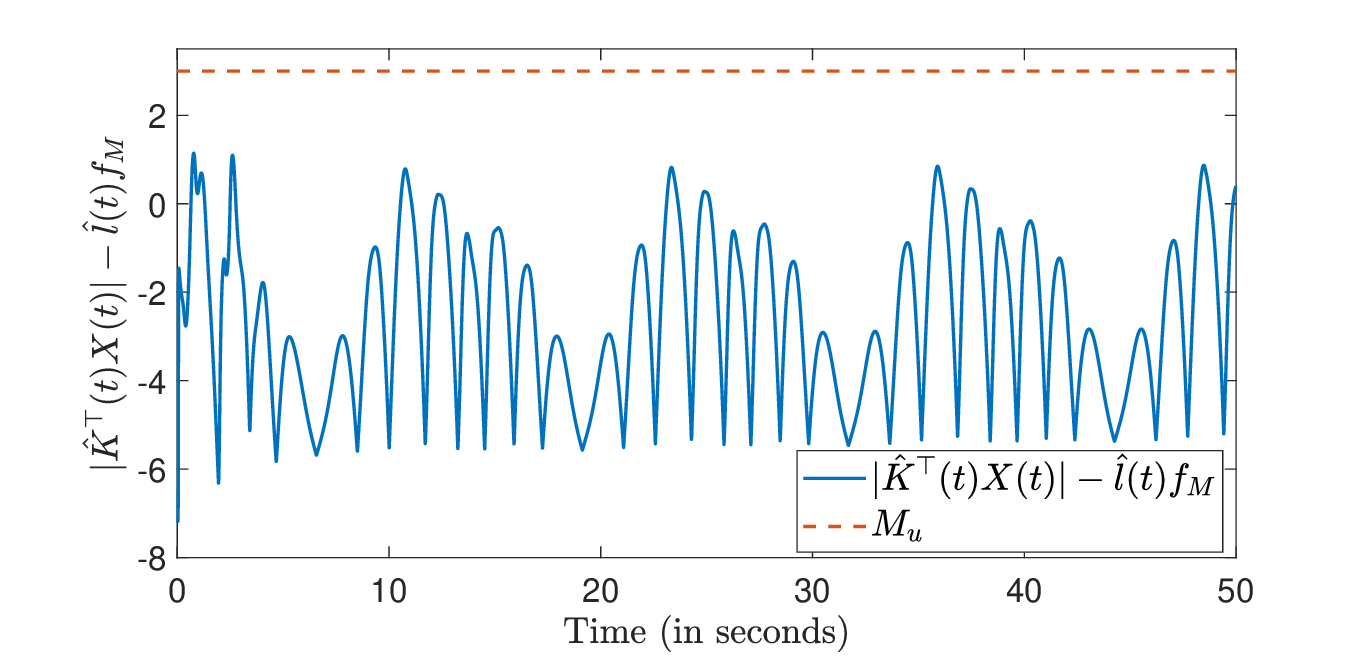}}}\hspace{5pt}
\caption{The stability condition in Assumption \ref{as:stability_cond} is satisfied for all time $t \geq 0$ when the value of $M_u$ is 3.} \label{Figure1}
\end{figure}
\begin{figure}[h!]
\centering{%
\resizebox*{9.3cm}{!}{\includegraphics{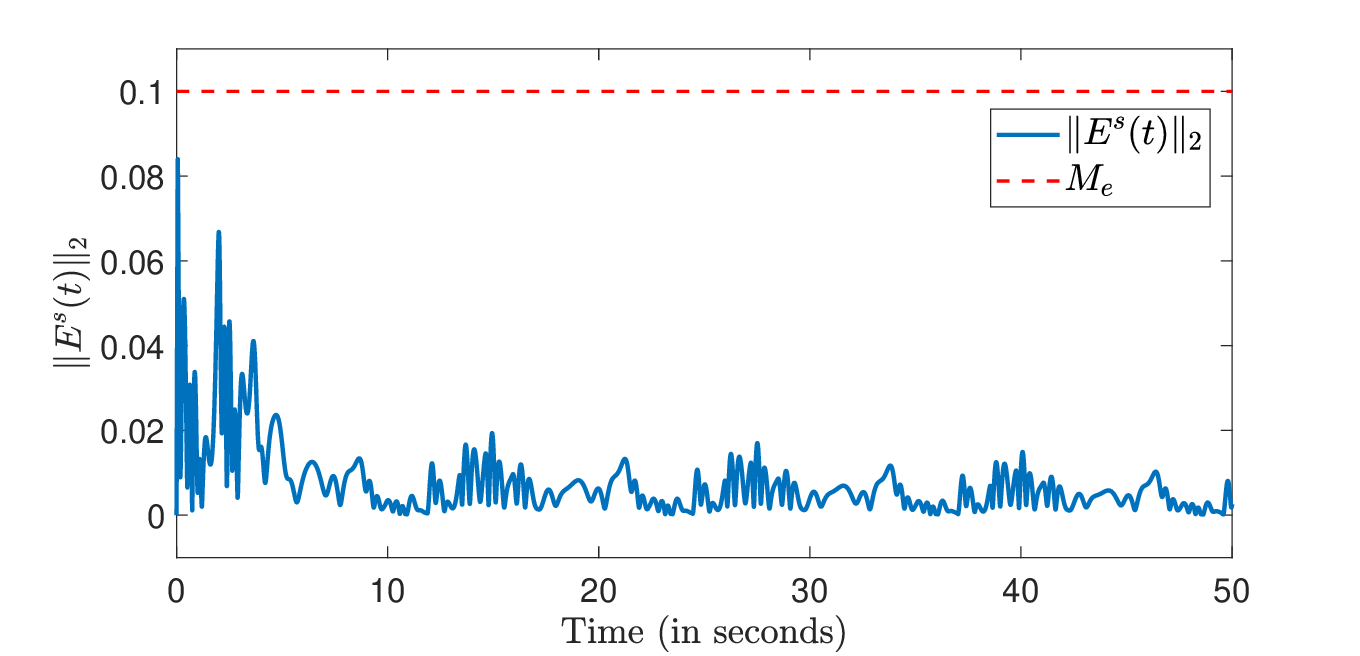}}}\hspace{5pt}
\caption{The modified tracking error $E^s(t) = X(t)-X_m^s(t)$ goes to zero as $t$ goes to infinity. Further, $\|E^s(t)\|_2 < M_e$ for all $t \geq 0$.} \label{Figure2}
\end{figure}
\begin{figure}[h!]
\centering{%
\resizebox*{9.3cm}{!}{\includegraphics{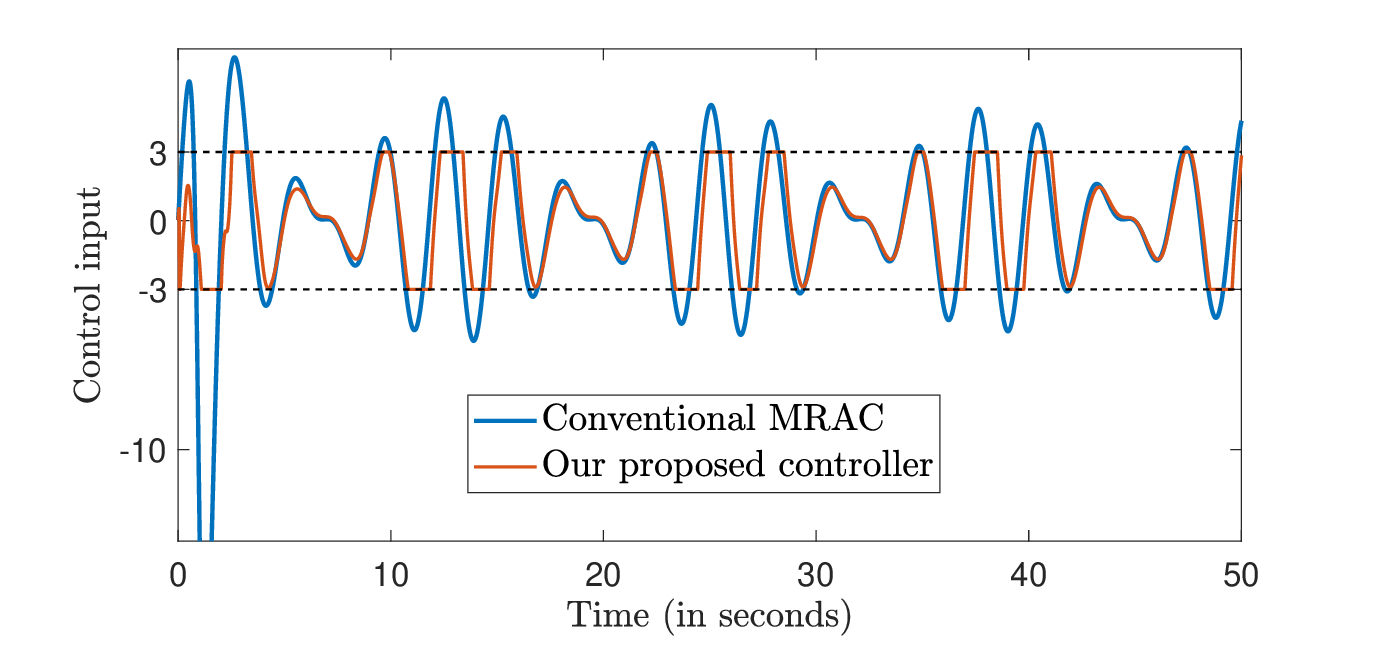}}}\hspace{5pt}
\caption{Our proposed control input $u^s$ (shown in red) given in \eqref{eq:mod_input} satisfies $|u^s(t)| \leq M_u= 3$ for all $t \geq 0$. However, the control input designed using a conventional MRAC approach  (shown in blue) fails to satisfy the input constraint.} \label{Figure3}
\end{figure}
\begin{figure}[h!]
\centering{%
\resizebox*{9.3cm}{!}{\includegraphics{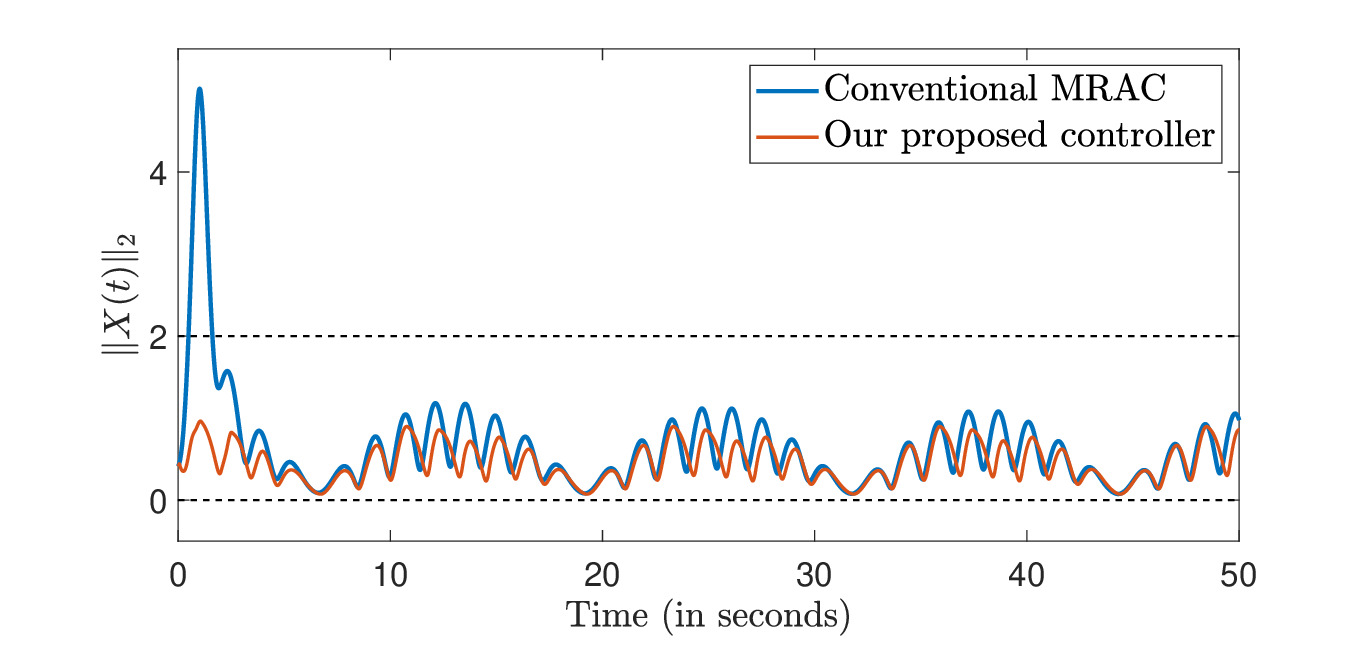}}}\hspace{5pt}
\caption{The state $X$ given in \eqref{eq:ex2_sys} satisfies $\|X(t)\|_2 < M_x= 2$ for all $t \geq 0$ when the control input is as in \eqref{eq:mod_input}. However, the plant state controlled using a conventional MRAC fails to satisfy the state constraint.} \label{Figure4}
\end{figure}
\begin{figure}[h!]
\centering{%
\resizebox*{9.3cm}{!}{\includegraphics{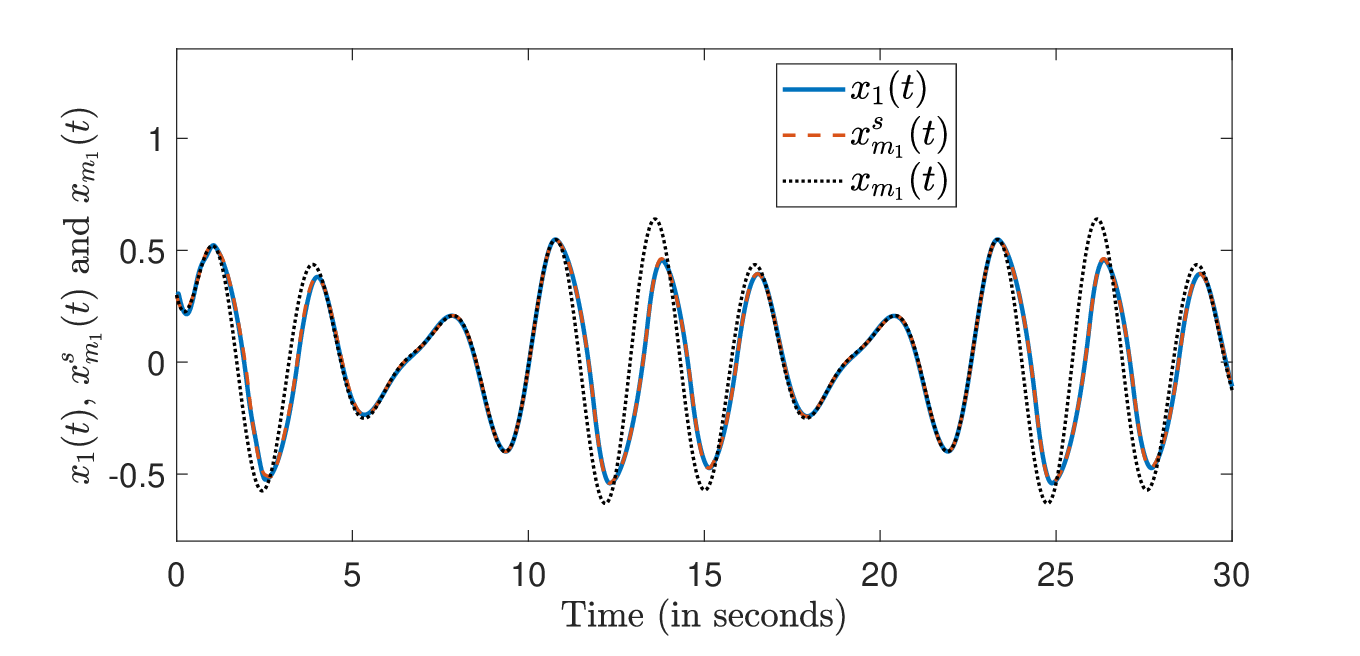}}}\hspace{5pt}
\caption{The state $x_1$ (shown in blue) tracks the modified target system state $x_{m_1}^s$ (shown in red). The deviation of $x_{m_1}^s(t)$ (red line) from $x_{m_1}(t)$ is due to the reference trajectory modification given in \eqref{eq:g}.} \label{Figure5}
\end{figure}
\begin{figure}[!htbp]
\centering{%
\resizebox*{9.3cm}{!}{\includegraphics{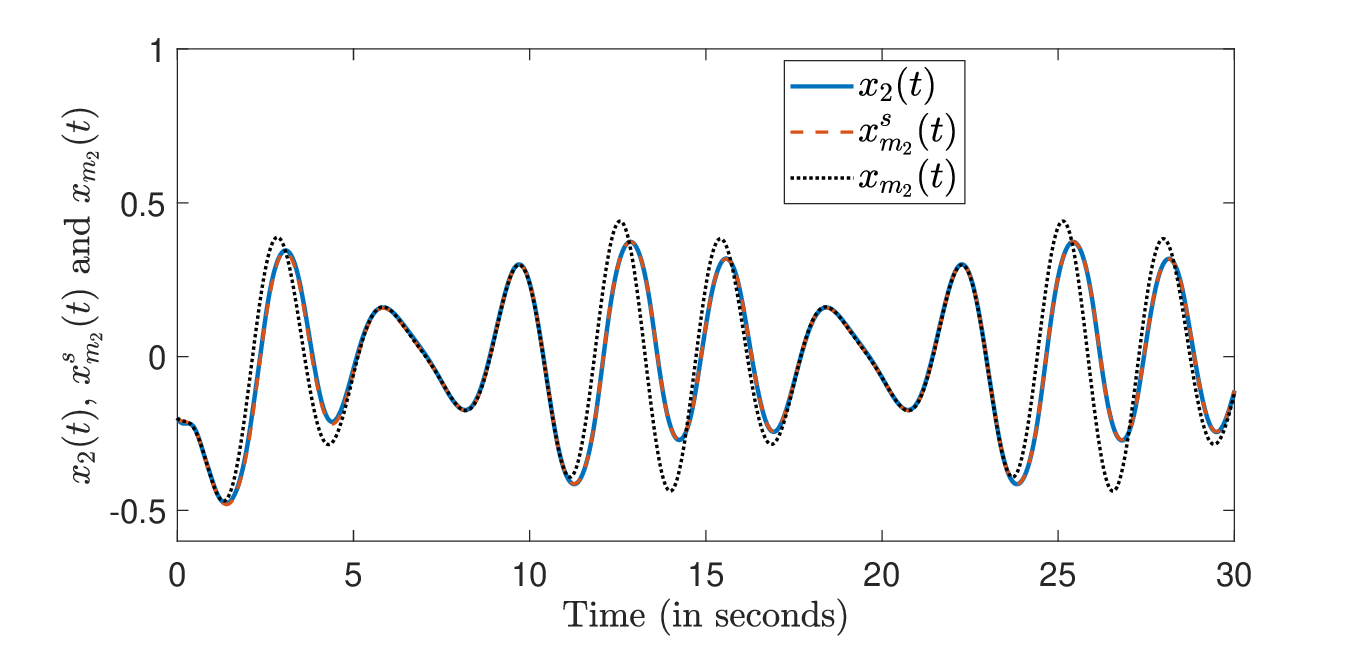}}}\hspace{5pt}
\caption{The state $x_2$ (shown in blue) tracks the modified target system state $x_{m_2}^s$ (shown in red). The deviation of $x_{m_2}^s(t)$ (red line) from $x_{m_2}(t)$ is due to the reference trajectory modification given in \eqref{eq:g}.} \label{Figure6}
\end{figure}
\begin{figure}[!htbp]
\centering{%
\resizebox*{9.3cm}{!}{\includegraphics{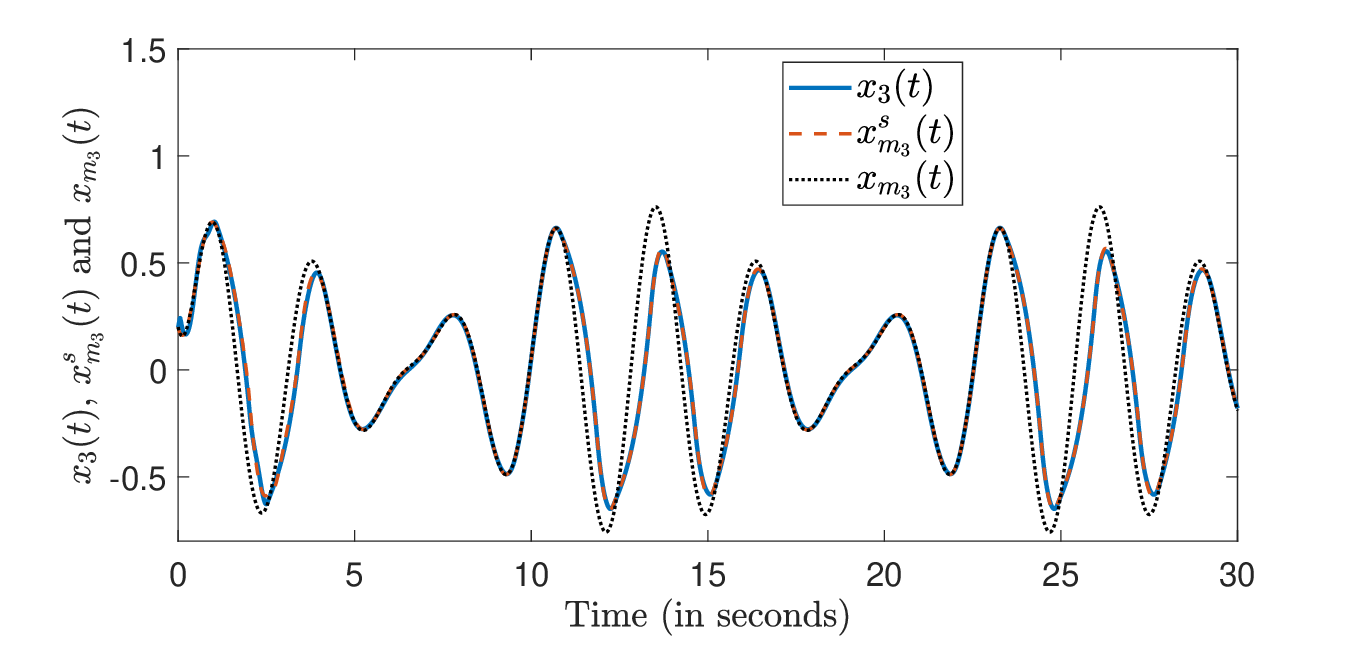}}}\hspace{5pt}
\caption{The state $x_3$ (shown in blue) tracks the modified target system state $x_{m_3}^s$ (shown in red). The deviation of $x_{m_3}^s(t)$ (red line) from $x_{m_3}(t)$ is due to the reference trajectory modification given in \eqref{eq:g}.} \label{Figure7}
\end{figure}
\begin{figure}[!htbp]
\centering{%
\resizebox*{9.3cm}{!}{\includegraphics{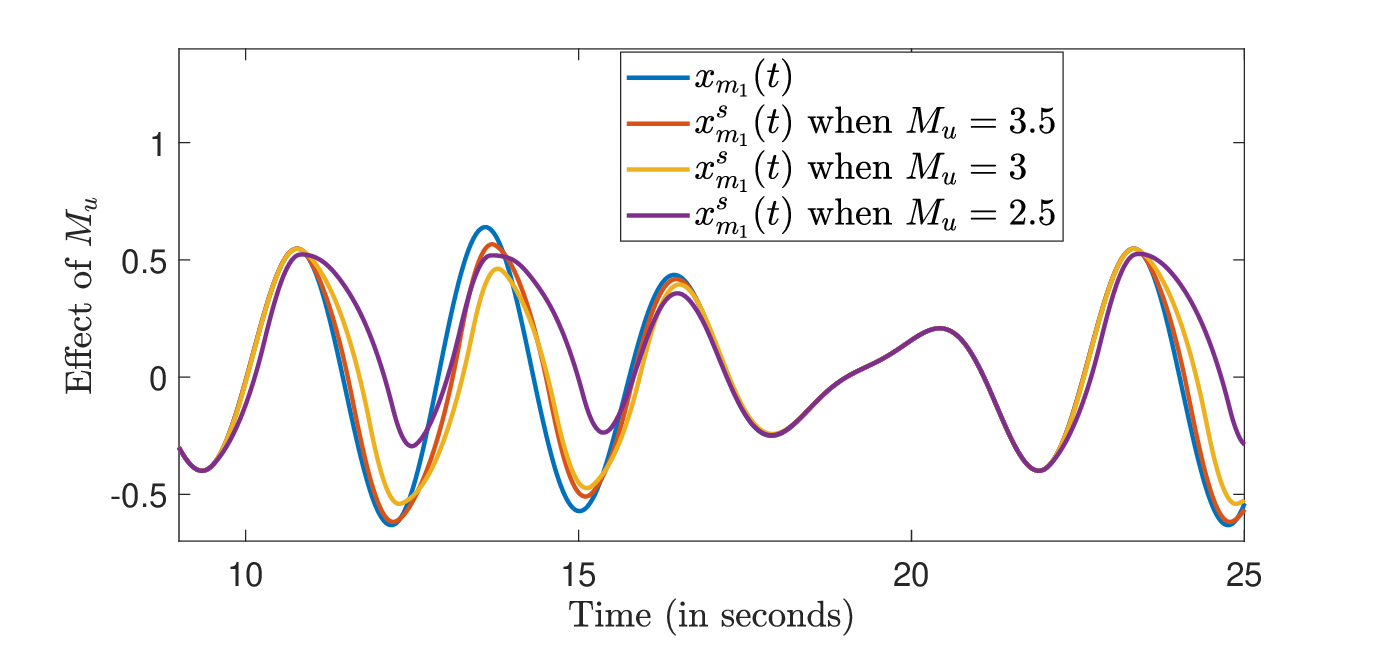}}}\hspace{5pt}
\caption{As $M_u$ increases, the effect of reference modification on $x_{m_1}$ reduces. When $M_u = 3.5$, the trajectory of $x_{m_1}^s$ is very close to the trajectory of the original target system state $x_{m_1}.$} \label{Figure8}
\end{figure}

Figure \ref{Figure1} shows that the stability condition in Assumption \ref{as:stability_cond} is satisfied for all time $t \geq 0$. Figure \ref{Figure2} shows that the modified tracking error $E^s(t) = X(t)-X_m^s(t)$ goes to zero as $t$ goes to infinity and $\|E^s(t)\|_2\leq M_e$ for all $t \geq 0$. Figures \ref{Figure3} and \ref{Figure4} show that the state and input constraints given in Section \ref{sec2} are satisfied for all $t \geq 0$ when the system in \eqref{eq:ex2_sys} is controlled using our proposed controller. Figures \ref{Figure5}-\ref{Figure7} show that the system state $X$ can closely track the target system state $X_m$. As $M_u$ increases, the effect of reference modification on the trajectory of $X_{m}$ should reduce. This is validated in Figure \ref{Figure8}. All these results validate the theoretical analysis done in Theorem \ref{th:state_ip_cons}. 

\textit{Note:} The difference between the blue trajectory in Figure \ref{Figure1} and the value of $M_u$ gives an indication of how much we can reduce $M_u$ without affecting the stability of the closed-loop dynamics of the plant. This knowledge is helpful in the case of practical safety-critical applications.

\section{Conclusions} \label{sec6}

In this paper, we have proposed a modified MRAC that can handle both state and input constraints in uncertain linear time-invariant systems. Our design approach for the modified MRAC involves using a barrier Lyapunov function to handle the state constraints and reference system modification to handle the input constraints. We have shown that under an easily verifiable stability condition, the system tracking error goes close to zero and both the state and input constraints remain satisfied for all time. We have also shown that, unlike existing works, our stability condition is given in terms of measurable or known quantities, making it possible to verify the condition even before implementing the control algorithm. We have demonstrated the efficacy of our control algorithm using a linear system example. Future work will focus on output feedback adaptive control under constraints.

\section{Acknowledgment}
Srikant Sukumar was partially sponsored for this research work by the Indo-French Center for Promotion of Advanced Research (IFCPAR) under collaborative grant 6001-1.

\bibliography{ref}

\end{document}